\documentclass[oneeqnum,onethmnum,onefignum,onetabnum]{siamltex1213}
\pdfoutput=1

\usepackage{color}
\usepackage[utf8]{inputenc}
\usepackage[english]{babel}

\usepackage{algorithm}
\usepackage{algorithmicx}
\usepackage{algpseudocode}
\usepackage{amsmath}
\usepackage{amssymb}
\usepackage{caption}
\usepackage{subcaption}
\usepackage{csvsimple}
\usepackage{enumerate}
\usepackage{graphicx}
\usepackage{textcomp}
\usepackage{verbatim}
\usepackage{pgfplots}
\usepackage{tikz}

\newcommand{\Hdiv}{$ H( \mathrm{div} ) $}
\newcommand{\Hcurl}{$ H( \mathrm{curl} ) $}

\title{A PARALLEL EDGE ORIENTATION ALGORITHM FOR QUADRILATERAL MESHES
\thanks{This work was supported by The Grantham Institute and the National Environmental Research Council [grant
numbers NE/K006789/1 and NE/K008951/1].
This work used the ARCHER UK National Supercomputing Service (http://www.archer.ac.uk).}}
\author{M.~Homolya\footnotemark[2]\ \footnotemark[3]
\and D.~A.~Ham\footnotemark[3]\ \footnotemark[4]}

\begin{document}

\maketitle

\renewcommand{\thefootnote}{\fnsymbol{footnote}}

\footnotetext[2]{The Grantham Institute, Imperial College London,
London, SW7 2AZ, UK}
\footnotetext[3]{Department of Computing, Imperial College London,
London, SW7 2AZ, UK}
\footnotetext[4]{Department of Mathematics, Imperial College London,
London, SW7 2AZ, UK}

\renewcommand{\thefootnote}{\arabic{footnote}}

\begin{abstract}
One approach to achieving correct finite element assembly is to ensure
that the local orientation of facets relative to each cell in the mesh
is consistent with the global orientation of that facet.
Rognes \textsl{et al.}\@ have shown how to achieve this for any mesh
composed of simplex elements, and deal.II contains
a serial algorithm to construct a consistent orientation of any
quadrilateral mesh of an orientable manifold.

The core contribution of this paper is the extension of this algorithm
for distributed memory parallel computers, which facilitates its seamless
application as part of a parallel simulation system.

Furthermore, our analysis establishes a link between the
well-known \emph{Union-Find} algorithm and the construction of a
consistent orientation of a quadrilateral mesh.
As a result, existing work on the parallelisation of the Union-Find
algorithm can be easily adapted to construct further parallel algorithms
for mesh orientations.
\end{abstract}



\pagestyle{myheadings}
\thispagestyle{plain}
\markboth{M.~HOMOLYA AND D.~A.~HAM}{PARALLEL QUADRILATERAL EDGE ORIENTATIONS}

\section{Introduction}

A characteristic of the finite element method for the solution of partial
differential equations is that the representation of functions over the
domain can be chosen from a wide range of function spaces.  This choice is
achieved by selecting a particular local function space to be used for each
cell in the meshed domain. Information is communicated between adjacent
cells either by shared degrees of freedom on the cell boundaries, or by flux
integrals over the facets\footnote{By \emph{facet} we mean a mesh entity of
  codimension 1. In two dimensions these are the cell edges, in three
  dimensions, the cell faces.} between cells. For every case except the
lowest degree continuous finite elements, these coupling terms require the
orientations of the two cells adjacent to each facet to be reconciled with
the orientation of that facet. Failure to do so will result in false
identification of degrees of freedom falling on the facets, or incorrect
accounting for the flux direction through facets.

Consequently, every finite element implementation which supports facet integrals, elements
of polynomial degree greater than one, or \Hdiv{} or \Hcurl{} conforming
elements must somehow ensure that adjacent cells agree on the orientation
of the intervening facet.  This can either be achieved by explicitly
recording orientation information and exploiting this in the local and/or
global assembly operations, or it can be achieved by ensuring, in a sense
which we will later make formal, that the local orientation of facets
relative to each cell in the mesh is consistent with the global orientation
of that facet.

The consistent numbering approach has the advantage that the integral
assembly code is simpler than that required by the explicit orientation
approach, at the cost of somehow establishing the consistent numbering and
with the limitation that a consistent numbering may not exist for all
meshes. Consistent numberings are simple to construct for any mesh composed
of simplex elements \cite{Rognes2012} but the same algorithm does not extend
to quadrilaterals.  Instead, the deal.II finite element library
\cite{Bangerth2007} contains an algorithm to construct a consistent
orientation of any quadrilateral mesh of an orientable
manifold\footnote{A preprint documenting the deal.II algorithm has
recently appeared \cite{agelek2015}.}.  The
algorithm employed in deal.II is serial. This is less than ideal on modern
supercomputers, for which it would be both more convenient and more
efficient to construct the orientation using a distributed parallel
algorithm. The core contribution of this paper is therefore to present a
distributed parallel extension of the algorithm in deal.II which creates a
global consistent numbering.  This algorithm is presented in
\S\ref{sec:parallel}, with experimental evaluation in \S\ref{sec:experiments}.
The formal definition of the problem is given in \S\ref{sec:statement}
and the proof of the existence of a solution in \S\ref{sec:analysis}.
\S\ref{sec:serial} restates the deal.II serial algorithm in terms of
this proof.  The parallel part of this
algorithm is, in fact, an adaption of the well-known \emph{Union-Find}
algorithm, so in \S\ref{sec:union-find} we place this algorithm in the context
of that existing work.  We conclude the paper in \S\ref{sec:conclusion}.

Facet orientation has been addressed in different ways by various finite
element packages. Rognes et al.\@ \cite{Rognes2012} give a more detailed
discussion of the features which make consistent orientations useful for
the efficient assembly of finite elements. They also provide an algorithm to
find consistent facet orientations for simplicial meshes, which is used in
FEniCS \cite{Hoffman2006} and in Firedrake \cite{Rathgeber2015}.  Other
finite element packages often use different approaches.  libMesh
\cite{Kirk2006} uses global vertex numbers to determine the direction of
edges for \Hcurl{} conforming elements, however neither \Hdiv{} conforming,
nor higher-order elements are supported.  FreeFem++ \cite{Hecht2012} does
not support quadrilaterals.  Nektar++ \cite{Cantwell2015} stores and uses
explicit per-cell edge orientations, and also supports mixed cell meshes
(e.g.\@ mixed triangles and quadrilaterals).  These orientations are
employed to match degrees of freedom on facets for higher-order elements,
however no sign flipping is necessary since \Hdiv{} and \Hcurl{} conforming
elements are not supported.  Raviart-Thomas elements in DUNE
\cite{Bastian2006} store explicit signs ($-1$ or $+1$) for each facet, which
are referenced from manually written formul\ae.  deal.II \cite{Bangerth2007}
employs the serial algorithm we describe in \S\ref{sec:serial}, as a
post-processing step after reading the mesh.

\section{Problem Statement}
\label{sec:statement}

This section aims to formally define the problem that we later
analyse and provide algorithmic solutions to.
Throughout this paper, we say that a mesh has \emph{consistent facet
  orientations}, if:
\begin{enumerate}[(i)]
  \item Facet orientation is \emph{intrinsic} to the facet, i.e.\@ it
    does not depend on the cell from which the facet is accessed.
  \item There is a fixed reference cell such that for each cell in the mesh,
    there exists a mapping to the reference cell under which the orientation
    of each facet incident to that cell matches the orientation of the
    corresponding facet of the reference cell.
\end{enumerate}
The orientation of facets on the reference cell can be arbitrary for
the purpose of the above definition,
although this choice can limit the scope of meshes for which a
consistent facet orientation exists.

In this paper, we tackle the problem of finding \emph{consistent
  edge orientations} for quadrilateral meshes on 2-dimensional
orientable manifolds.

\section{Proof of the Existence of a Solution}
\label{sec:analysis}

We prove that \emph{consistent edge orientations} exist for any
quadrilateral mesh on an \emph{orientable} 2-dimensional manifold.
Later we build on the analysis below to construct algorithms that
find consistent orientations.

As mentioned earlier, the choice of edge orientations on the
reference quadrilateral can affect the set of meshes for which
consistent orientations exist.
Up to  symmetry, there are four different edge orientations of a
quadrilateral. These are shown in Figure~\ref{fig:four_quads}. The
symmetries arise since
we can arbitrarily choose the first vertex and the direction of
traversal when defining a mapping to the reference cell.
For each case, one can find a mesh on a 2-dimensional manifold which
does not have consistent edge orientations. \cite{Bangerth} includes a
counter-example for (a). A structured grid with odd number of cells in
both direction, folded into a torus, is a counter-example for (b) and
(d). A structured grid bent into a Möbius strip is a counter-example
for (c). However, when we restrict our attention to meshes on
\emph{orientable} manifolds (that is, the vast majority of domains actually
employed for finite element problems), then only (c) remains without a
counter-example. We here prove existence for this case:

\begin{figure}
  \centering

  \subcaptionbox{}
    {\includegraphics{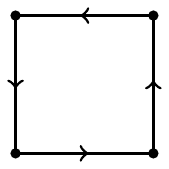}}
  \quad
  \subcaptionbox{}
    {\includegraphics{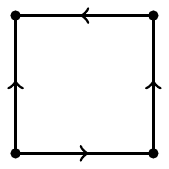}}
  \quad
  \subcaptionbox{\label{fig:quad3}}
    {\includegraphics{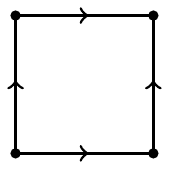}}
  \quad
  \subcaptionbox{}
    {\includegraphics{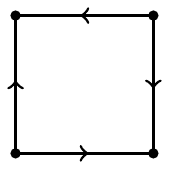}}

  \caption{Possible different quadrilateral edge orientations, after
    considering symmetries.}
  \label{fig:four_quads}
\end{figure}

\begin{theorem}
Consistent edge orientations exist for any quadrilateral mesh on an
orientable 2-dimensional manifold, given the edge orientations on the
reference quadrilateral as in Figure~\ref{fig:quad3}.
\end{theorem}

\begin{proof}
  Notice that on the reference quadrilateral, opposite edges always have the
  same orientation. When constructing a global orientation of the cells in a
  mesh, this implies that fixing the orientation of one edge
  \emph{determines} the orientation of the opposite edge. Conversely, the
  orientation of an edge imposes no restriction on the orientation of the
  two adjacent edges. Since each interior edge of the mesh is adjacent to
  two cells, fixing the orientation of one edge actually fixes the
  orientation of every edge reachable from the first edge by an unbroken
  sequence of opposite cell edges. This defines a relation ``the orientation
  of edge $a$ determines the orientation of edge $b$''.  We will refer to
  this as the \emph{orientation determination relation}. It is easy to see
  that this relation is reflexive, symmetric and transitive, and is
  therefore an equivalence relation which defines equivalence classes on
  the set of edges. The form of each equivalence class is a path or ribbon
  (see Figure~\ref{fig:ribbons})
  through the mesh connecting successive opposite faces. Within each class,
  the orientation of any edge determines the orientation of all other
  edges. Conversely, since the orientation of an edge only determines the
  orientation of the opposite edge and not that of the adjacent edges,
  separate classes can be oriented independently.

\begin{figure}
  \centering
  \includegraphics[width=0.7\textwidth]{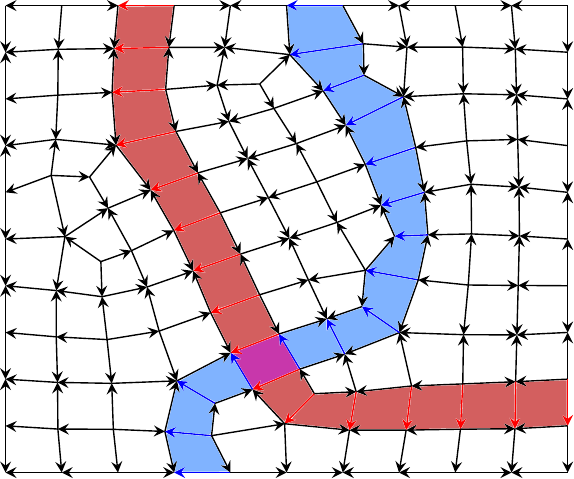}

  \caption{Unstructured quadrilateral mesh with consistent edge
    orientations. Two ``ribbons'' (sets of edges where the orientation
    of any edge determines the orientation of all edges) are
    highlighted in dark red and light blue.}
  \label{fig:ribbons}
\end{figure}

Note that if there exists a consistent edge orientation for a quadrilateral
mesh, one can invert the orientation of every edge in any equivalence class,
and still have a consistent edge orientation. Thus when assigning
orientations to a set of dependent edges, we can arbitrarily choose the
orientation of one edge, and that orientation propagates to all dependent
edges. Either both initial choices are fine, or a consistent orientation
does not exist.

This implies that a consistent orientation will always exist for a mesh,
unless there is a set of dependent edges for which the orientations necessarily
conflict. That is, there is an edge $u$, whose orientation forces a
different orientation on an edge $v$ through path $ p_1 $ than through path
$ p_2 $. From a slightly different point of view, there is a cycle $ p = u
\overset{p_1}{\leadsto} v \overset{p'_2}{\leadsto} u $ around $u$, where $
p'_2 $ is the reverse of $ p_2 $, such that the orientation of $u$ forces a
conflicting orientation on itself through path $p$. Without loss of
generality we can assume that $p$ is a simple cycle, i.e.\@ $p$ contains
each of its edges only once (except for $u$, which appears as starting and
ending edge in $p$). In this case the quadrilateral cells in $p$ form a
Möbius strip (see Figure~\ref{fig:moebius_strip}), which implies that the
mesh is not on an orientable manifold.

\begin{figure}
  \centering
  \includegraphics{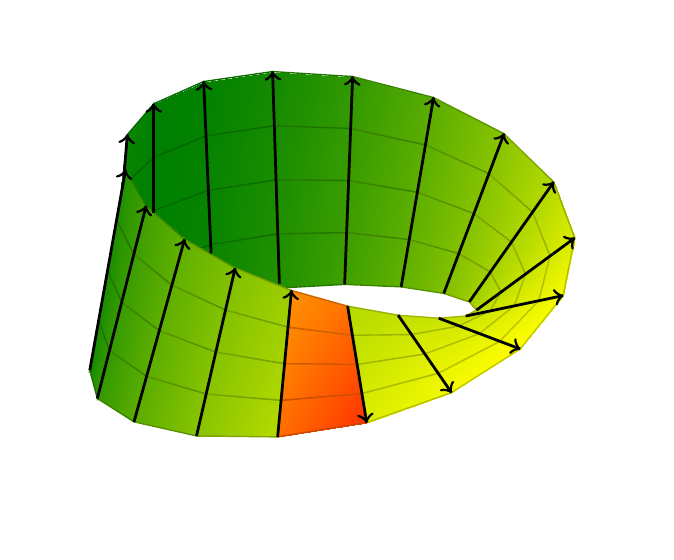}

  \caption{Edge orientations propagate in a conflicting manner on
    a Möbius strip, and \textit{vice versa}, if a cycle propagates
    orientations in a conflicting way, then connecting the cells will
    form a Möbius strip.}
  \label{fig:moebius_strip}
\end{figure}

We have shown
that if the mesh does not admit consistent edge orientations, then its
manifold is non-orientable. It follows immediately that consistent orientations
exist for any quadrilateral mesh on an orientable manifold.
\end{proof}

\section{Serial Algorithm}
\label{sec:serial}

Based on the above discussion, Alg.~\ref{alg:serial} describes
the steps of consistently orienting a quadrilateral mesh. We assume
that initially $ \mathtt{visited}[e] = \mathtt{False} $ and
$ \mathtt{orientations}[e] $
is undefined for all $ e \in E $, where $ E $ denotes the set
of edges in the mesh. For generality, lines 13 and 14 lack the details
of mesh representation and orientation representation.

\begin{algorithm}
  \caption{Serial algorithm}
  \label{alg:serial}
  \begin{algorithmic}[1]
    \ForAll {$ e \in E $}
      \State \Call{Orient}{$ e $, $ \uparrow $}
    \EndFor
    \Statex
    \Procedure {Orient}{$ e $, $ o $}
    \If {$ \mathtt{visited}[e] $}
      \If {$ \mathtt{orientations}[e] \neq o $}
      \Comment{Möbius strip found}
      \State \textbf{abort}
      \EndIf
    \Else
      \State $ \mathtt{visited}[e] \gets \mathtt{True} $
      \State $ \mathtt{orientations}[e] \gets o $
      \ForAll {$ c \in \text{cells incident to } e $}
        \State $ e' \gets \text{edge opposite to } e \text{ in } c $
        \State $ o' \gets \text{orientation for } e' \text{ consistent with } e \text{ having orientation } o $
        \State \Call {Orient}{$ e' $, $ o' $}
      \EndFor
    \EndIf
    \EndProcedure
  \end{algorithmic}
\end{algorithm}

Lines 10--16 run only $ |E| $ times, since each time they flip
a \texttt{False} to \texttt{True} in $ \mathtt{visited} $. Since the
quadrilateral mesh is on a 2-dimensional manifold, only 1 or 2 cells
can be incident to an edge (line 12), therefore lines 10--16 take
$ O(|E|) $ total runtime, not considering the time spent in recursive
\textsc{Orient} calls outside these lines. It also follows that
\textsc{Orient} is called $ O(|E|) $ times: $ |E| $ times in line 2,
and between $ |E| $ and $ 2 |E| $ times in line 15. Hence the serial
algorithm takes $ O(|E|) $ time.

The serial algorithm was first implemented in deal.II
\cite{Bangerth2007}, and briefly described in its
documentation \cite{Bangerth}.

\section{Parallel Algorithm}
\label{sec:parallel}

This section describes the \emph{parallel} extension of the above
algorithm for \emph{distributed memory} systems, as it is
implemented in Firedrake.

The basic idea of Firedrake's parallel implementation is to first
solve edge orientation locally, then flip some local ribbon segments
until all local domains agree on the orientation of all shared edges
(see Figure~\ref{fig:ribbon_segments}).
To speed up flipping local ribbon segments, each parallel process
records connections between its shared edges. Interior edges are
re-aligned only when agreement on the orientations of all shared
edges is reached.

\begin{figure}
  \centering
  \includegraphics[width=0.6\textwidth]{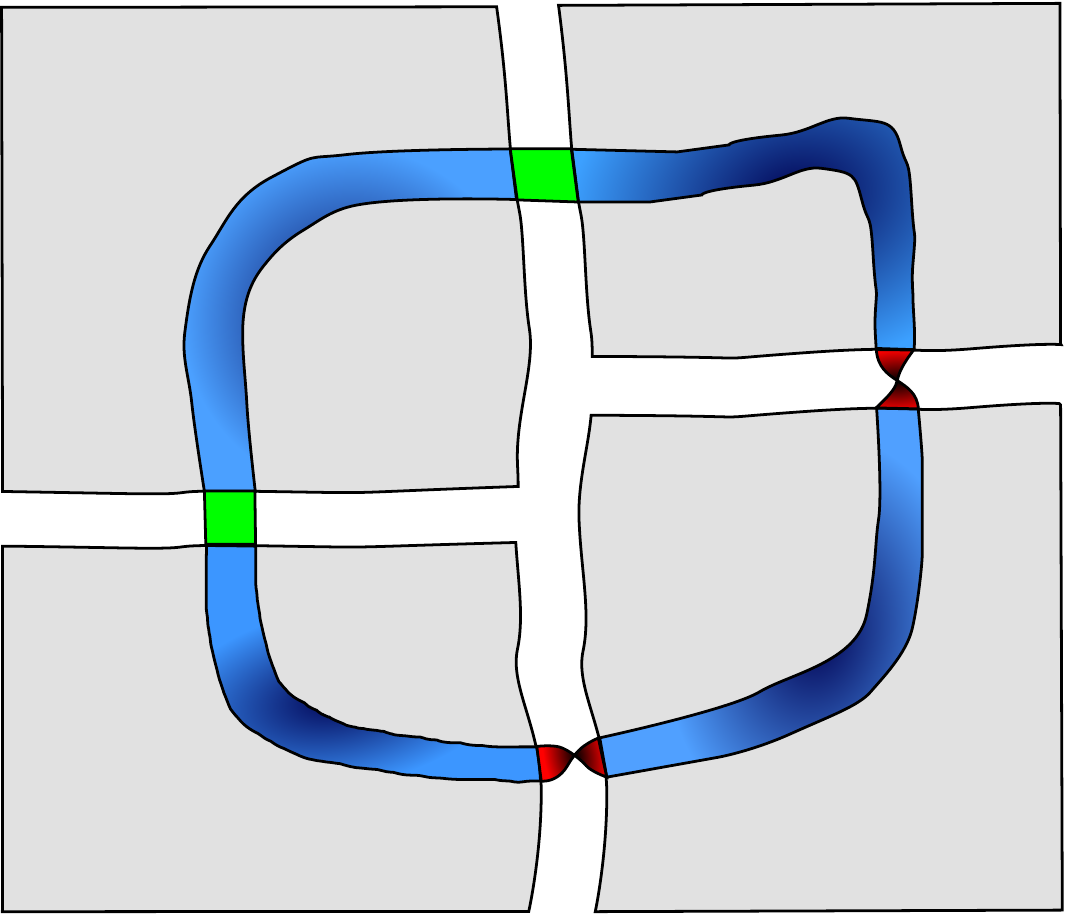}

  \caption{A ``ribbon'' in the global mesh, and its segments in the
    local domains. It is assumed that the orientation of edges is
    consistent inside these segments. However, the local domains
    \emph{may} (green, straight connection) or \emph{may not}
    (red, twisted connection) agree on the
    orientation of shared edges.}
  \label{fig:ribbon_segments}
\end{figure}

Alg.~\ref{alg:parallel} describes how parallel processes negotiate the
orientation of shared edges. Before the negotiation starts, each
process aligns its edges locally, while populating the
\texttt{affects\_edge} and \texttt{affects\_orient} arrays.  If $ u $
and $ v $ are edges shared with other processes, and connected by a
local ribbon segment, then \texttt{affects\_edge} is populated with
entries $ \mathtt{affects\_edge}[u] = v $ and $
\mathtt{affects\_edge}[v] = u $, and \texttt{affects\_orient} is
populated with entries $ \mathtt{affects\_orient}[u] =
\mathtt{affects\_orient}[v] ={} \uparrow $ when $ u $ and $ v $ have
the same orientation, and with entries $ \mathtt{affects\_orient}[u] =
\mathtt{affects\_orient}[v] ={} \downarrow $ otherwise. If $ u $ is a
shared edge, but it does not connect to any other shared edge, these
arrays are not updated. If we encounter a closed loop while locally
orienting edges, there is no need to negotiate the orientation its
edges with other parallel processes.

Another prerequisite is to initialise the \texttt{our\_weight} and
\texttt{our\_orient} arrays with the weights and proposed orientations
of the shared edges of a parallel process. The weights shall be
globally unique for all local ribbon segments. We use $
\mathtt{our\_weight}[e] = \mathit{size} \cdot l(e) + \mathit{rank} $,
where $ l(e) $ is the local index of the ribbon segment which $ e $
belongs to, \textit{size} is the number of parallel process, and
\textit{rank} is the index of the current process. We need to ensure
that locally connected edges have the same weight and consistent
orientation initially. This invariant is maintained by
Alg.~\ref{alg:parallel}.

\begin{algorithm}
  \caption{Parallel algorithm}
  \label{alg:parallel}
  \begin{algorithmic}[1]
    \State $ \mathtt{conflict} \gets \mathtt{True} $
    \While {\texttt{conflict}}
      \State \Call {ExchangeEdgeData}{\texttt{our\_orient}, \texttt{their\_orient}}
      \State \Call {ExchangeEdgeData}{\texttt{our\_weight}, \texttt{their\_weight}}
      \State $ \mathtt{conflict} \gets \mathtt{False} $
      \ForAll {$ e \in \text{shared edges} $}
        \If {$ \mathtt{our\_weight}[e] = \mathtt{their\_weight}[e]$}
          \If {$ \mathtt{our\_orient}[e] \neq \mathtt{their\_orient}[e] $}
            \Comment{Möbius strip found}
            \State \textbf{abort}
          \EndIf
        \ElsIf {$ \mathtt{our\_weight}[e] < \mathtt{their\_weight}[e]$}
          \State $ \mathtt{our\_weight}[e] = \mathtt{their\_weight}[e] $
          \State $ \mathtt{our\_orient}[e] = \mathtt{their\_orient}[e] $
          \If {$ e \in \mathtt{affects\_edge} $}
            \State $ e' \gets \mathtt{affects\_edge}[e] $
            \State $ o' \gets \mathtt{affects\_orient}[e] \text{ XOR } \mathtt{our\_orient}[e] $
            \If {$ \mathtt{our\_orient}[e'] \neq o' $}
              \State $ \mathtt{conflict} \gets \mathtt{True} $
            \EndIf
            \State $ \mathtt{our\_weight}[e'] = \mathtt{our\_weight}[e] $
            \State $ \mathtt{our\_orient}[e'] = o' $
          \EndIf
        \EndIf
      \EndFor
      \State \texttt{conflict} $ \gets $ \Call{AllReduce}{\texttt{conflict}, \textsc{Or}}
    \EndWhile
  \end{algorithmic}
\end{algorithm}

Alg.~\ref{alg:parallel} consists of a main loop which terminates when there
are no more conflicts in the proposed orientations for shared edges. In
line~3--4 each parallel process exchanges its proposed orientations
and weights for its shared edges with its
neighbours. \texttt{their\_weight} and \texttt{their\_orient} have
type and size identical to \texttt{our\_weight} and
\texttt{our\_orient} respectively, but they contain the values
proposed by the
neighbouring processes. The general rule is to adopt the orientation with
the larger weight.  Lines~11--23 handle the adoption of the remote weight and
orientation. Line~14 checks if such an update needs propagation via the
local ribbon segment. Lines~15--21 propagate the new information to the
other connected shared edge, and check if it causes that edge to flip.
If yes, then
another round of exchange is necessary, since the orientation for $ e' $
that this process forwarded in line~3, has now become outdated. When the
local weight is larger than the remote weight, there is nothing to do: the
other parallel process will adopt our orientation.  Since all weights are
initially unique, and then they are copied between neighbours, the local
weight being equal to the remote weight means that both parallel processes
have their local ribbon segment aligned to the same ``source''. Therefore a
mismatch in the orientation is only possible if the global mesh contains a
Möbius strip, which is checked in lines~8--10.


This algorithm terminates in at most $ k $ communication rounds, where
$ k $ is the maximum number of local segments in any ribbon of the
global mesh. Since each ribbon is ``untwisted'' independently, further
discussion focuses on one arbitrary ribbon in the global mesh.
Any time there is a conflict between the orientation of two local
segments, the one with the higher weight ``wins'', thus we can
conclude that the orientation of a ribbon in the global mesh is
ultimately determined by the local segment which has the highest
weight. Initially only one segment has this highest weight. In each
communication round one or two additional segments adopt this weight
and align their orientations, until the information has propagated
to all segments. At that point, there are either no conflicts along
the ribbon, or the algorithm will abort in the next communication
round due to finding a Möbius strip. Therefore the parallel
algorithm terminates in at most $ k $ communication rounds.

Theoretically, $ k $ can be large with respect to other parameters of the
mesh. For example, one can mesh an annulus domain with a very long,
spiraling ribbon. However, a far more typical scenario is for each ribbon to
either enter one side of the domain and exit the other, or to form a
loop. In either of these cases, assuming the domain has bounded aspect ratio
and that the parallel decomposition approximately minimises the surface area
of partitions, $ k $ will be $ O(\sqrt{P}) $, where $ P $ is the number of
parallel processes.

\section{Experiments}
\label{sec:experiments}

We have run several scaling experiments on the ARCHER UK National
Supercomputing Service. Our experimentation framework is available as
\cite{miklos_homolya_2015_31190}. We used 6 different meshes to
evaluate the performance of the parallel algorithm described in
\S\ref{sec:parallel}:
\begin{description}
  \item[\texttt{s\_square}:] structured grid on a square domain
  \item[\texttt{s\_sphere}:] cubed sphere mesh (a quadrilateral
    mesh of the surface of a sphere formed by deforming a cube to sphere shape
    and refining)
  \item[\texttt{u\_square}:] unstructured mesh on a square domain
  \item[\texttt{u\_sphere}:] unstructured mesh on the surface of a sphere
  \item[\texttt{t10}, \texttt{t11}:] unstructured meshes with
    non-uniform resolution (see Figure~\ref{fig:t10_t11})
\end{description}
The first two were generated using Firedrake utility functions, the other
four were generated with Gmsh \cite{Geuzaine2009}, version
2.8.3. \texttt{t10} and \texttt{t11} are examples from the Gmsh tutorial. To
give an impression of the level of irregularity of these meshes,
low-resolution versions are presented in Figure~\ref{fig:t10_t11}.
We used meshes with cell count in the order of millions, unfortunately mesh
generation becomes very difficult for meshes bigger than that.

\begin{figure}
  \centering
  \includegraphics[height=1.5in]{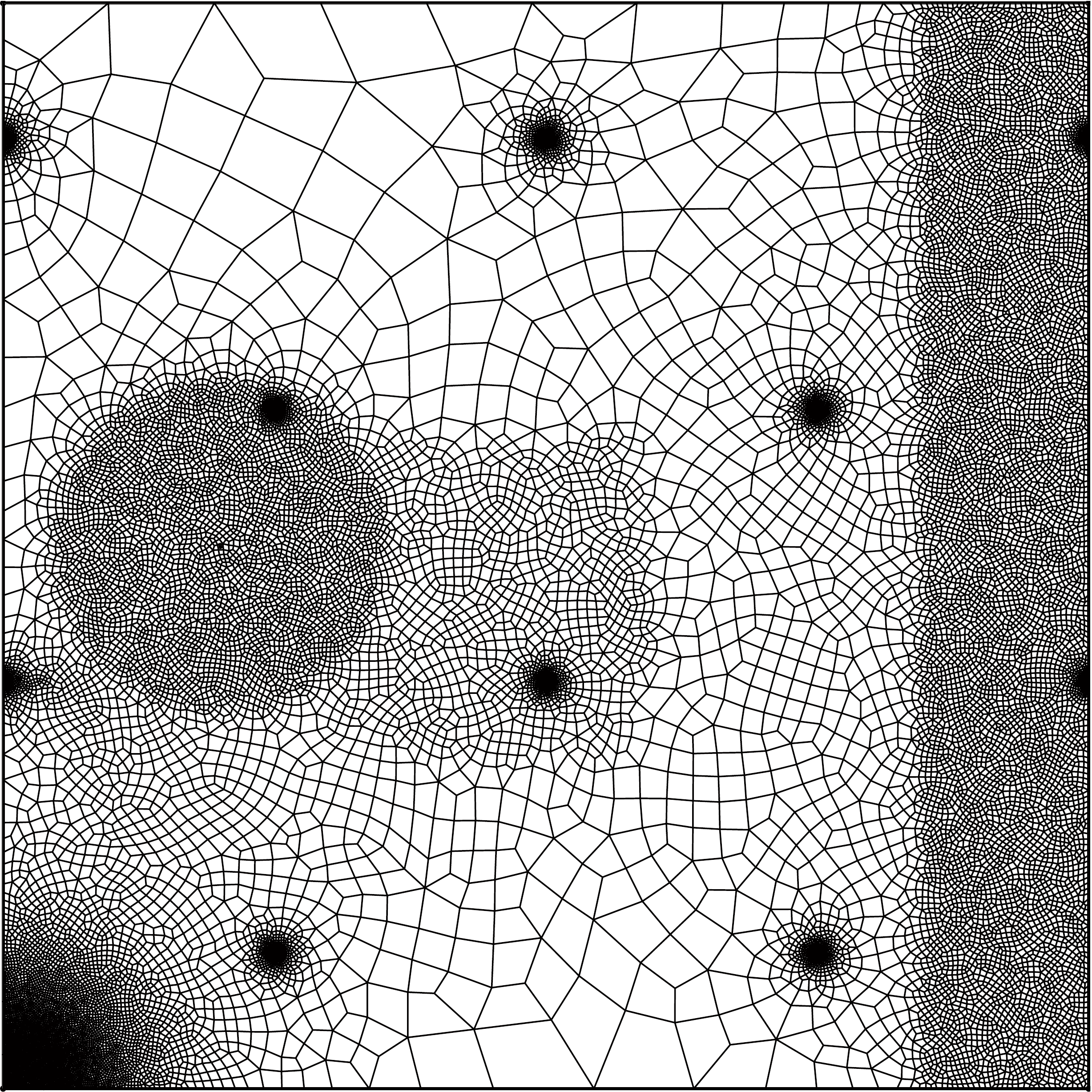}
  \quad
  \includegraphics[height=1.5in]{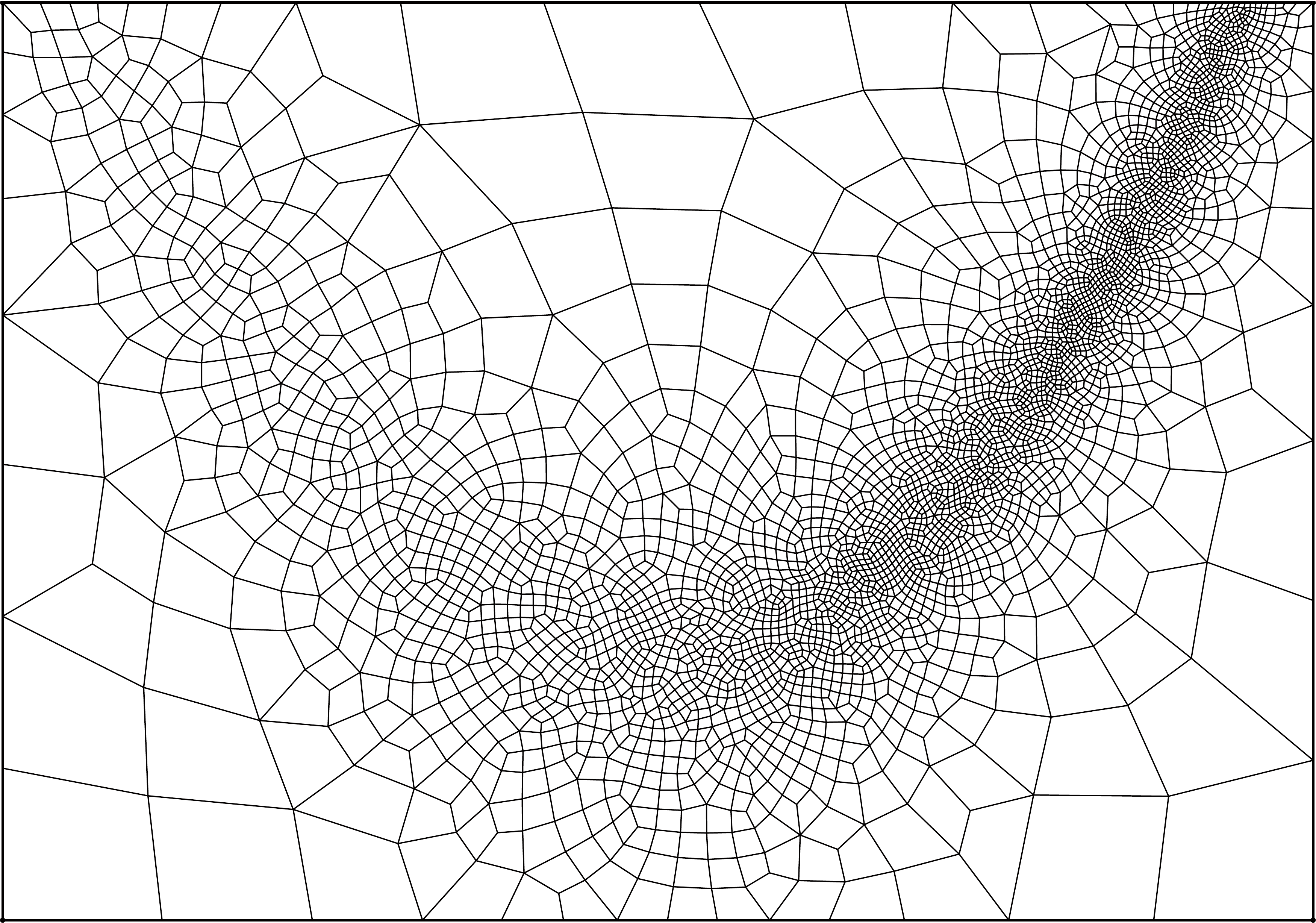}
  \caption{Low-resolution version of mesh samples \texttt{t10} (left) and \texttt{t11} (right).}
  \label{fig:t10_t11}
\end{figure}

We used Firedrake to construct consistent edge orientations for each
of these meshes on up to 24576 cores (1024 computing nodes, each with
24 CPU cores). We measured the number of required communication rounds
(Figure~\ref{fig:comm_rounds}) and execution time
(Figure~\ref{fig:exec_time}). Each experiment was repeated 4 times,
except for those involving the largest number of cores: for cost
efficiency reasons, they were carried out only once. The number of
communication rounds were the same during each instantiation of the
same experiment, however, the execution times varied greatly,
especially when they were under a second. Therefore reported execution
times should be taken with a grain of salt. Figure~\ref{fig:exec_time}
shows the average of 4 measurements for each data point.

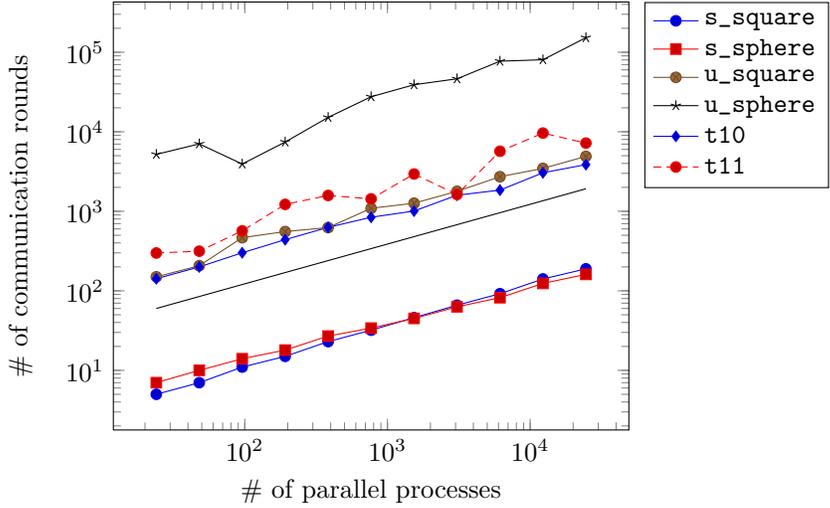
\begin{figure}
  \centering
  \begin{tikzpicture}
    \begin{loglogaxis}[
      legend cell align=left,
      legend pos=outer north east,
      xlabel={\# of parallel processes},
      ylabel={\# of communication rounds},
      ]
      \addplot table [x=P, y=rounds, col sep=comma] {rounds_s_square.csv};
      \addlegendentry{\texttt{s\char`_square}}
      \addplot table [x=P, y=rounds, col sep=comma] {rounds_s_sphere.csv};
      \addlegendentry{\texttt{s\char`_sphere}}
      \addplot table [x=P, y=rounds, col sep=comma] {rounds_u_square.csv};
      \addlegendentry{\texttt{u\char`_square}}
      \addplot table [x=P, y=rounds, col sep=comma] {rounds_u_sphere.csv};
      \addlegendentry{\texttt{u\char`_sphere}}
      \addplot table [x=P, y=rounds, col sep=comma] {rounds_t10.csv};
      \addlegendentry{\texttt{t10}}
      \addplot table [x=P, y=rounds, col sep=comma] {rounds_t11.csv};
      \addlegendentry{\texttt{t11}}

      \addplot [black] table [x=P, y=rounds] {
        P     rounds
        24    60
        24576 1920
      } node[pos=0.4,pin=355:{$\sqrt{P}$}] {};
    \end{loglogaxis}
  \end{tikzpicture}
  \caption{Number of required communication rounds for each mesh sample.}
  \label{fig:comm_rounds}
\end{figure}

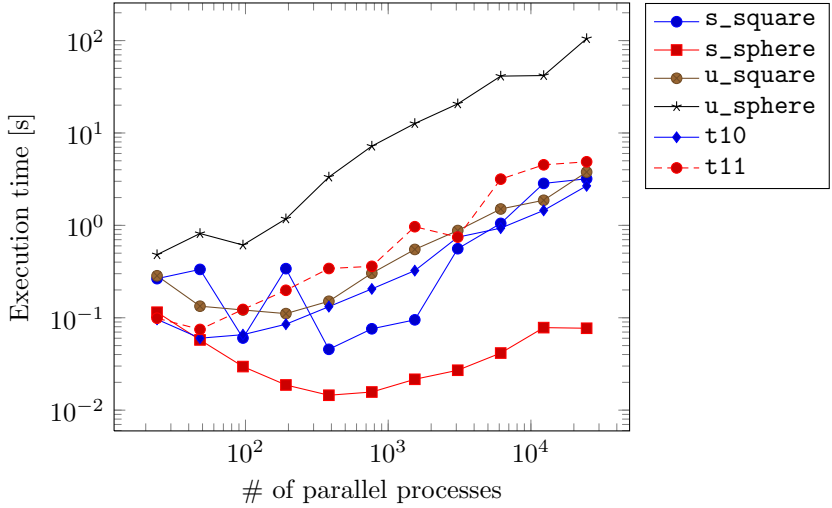
\begin{figure}
  \centering
  \begin{tikzpicture}
    \begin{loglogaxis}[
      legend cell align=left,
      legend pos=outer north east,
      xlabel={\# of parallel processes},
      ylabel={Execution time [s]},
      ]
      \addplot table [x=P, y=time, col sep=comma] {time_s_square.csv};
      \addlegendentry{\texttt{s\char`_square}}
      \addplot table [x=P, y=time, col sep=comma] {time_s_sphere.csv};
      \addlegendentry{\texttt{s\char`_sphere}}
      \addplot table [x=P, y=time, col sep=comma] {time_u_square.csv};
      \addlegendentry{\texttt{u\char`_square}}
      \addplot table [x=P, y=time, col sep=comma] {time_u_sphere.csv};
      \addlegendentry{\texttt{u\char`_sphere}}
      \addplot table [x=P, y=time, col sep=comma] {time_t10.csv};
      \addlegendentry{\texttt{t10}}
      \addplot table [x=P, y=time, col sep=comma] {time_t11.csv};
      \addlegendentry{\texttt{t11}}
    \end{loglogaxis}
  \end{tikzpicture}
  \caption{Execution time to construct consistent edge orientations in
    parallel for each mesh sample.}
  \label{fig:exec_time}
\end{figure}

Unsurprisingly, the construction edge orientations for structured
meshes takes relatively few communication rounds.
Figure~\ref{fig:comm_rounds} also shows that for a rectangular domain,
the number of required communication rounds stays the same order of
magnitude in case of non-uniform meshing as with uniform resolution.
However, the unstructured mesh on a spherical domain requires at least
one order of
magnitude more communication rounds. This is probably caused by the
fact that a sphere does not have a boundary, thus all ribbons must
form a closed loop, so traversing a ribbon must continue until the
starting edge is reached again. Regardless, our measurements confirm
our predicted $ O(\sqrt{P}) $ communication rounds for strong
scaling. Table~\ref{tab:slopes} collects the slopes of lines fitted to
the log-log data shown in Figure~\ref{fig:comm_rounds}, all values
being close to the expected $ 0.5 $.

\begin{table}
  \centering
  \begin{filecontents*}{slopes.csv}
mesh,slope
u\_square,0.485698
s\_square,0.527252
t10,0.470928
t11,0.494417
s\_sphere,0.446810
u\_sphere,0.526821
\end{filecontents*}
\csvreader[tabular=l | l, table head=\textbf{Test case} & \textbf{Slope} \\ \hline]
  {slopes.csv}{}{\texttt{\csvcoli} & \csvcolii}
  \caption{Slopes of lines fitted to the log-log data shown in Figure~\ref{fig:comm_rounds}.}
  \label{tab:slopes}
\end{table}

In most of the cases shown in Figure~\ref{fig:exec_time}, the execution time
initially decreases as the number of parallel processes increases. It is
likely that this is due to the parallelisation of the local work each
process has to do. Eventually the increasing number of communication
rounds becomes dominant, and the execution times grow with the number
parallel processes. However, even the unstructured sphere mesh with
almost one million cells, running on about a quarter of ARCHER, takes only
105 seconds. For users running supercomputer-scale simulations, the
few minutes spent on edge orientations are not likely to become a
bottleneck.

\section{Alternative Approaches to Parallelism}
\label{sec:union-find}

The edge orientation problem can be reduced to the well-known
\emph{Union-Find algorithm}, also known as the \emph{disjoint set} data
structure \cite{Cormen2001}.  We established in \S\ref{sec:analysis}
that the orientation determination relation is an equivalence relation
which defines equivalence classes (the ``ribbons'') on the set of edges.
To discover the ribbons, we can immediately use the Union-Find
algorithm. However, with an extension to it, we can also retrieve the
correct edge orientations.

First we briefly introduce the Union-Find algorithm. Let $ e_1,
e_2, \ldots, e_n $ be abstract elements, and $ S_1, S_2, \ldots, S_k $
disjoint non-empty sets over these elements, with $k \le n$. The sets are
initially singletons: $ S_i = \{ e_i \}, i = 1\ldots{}n $.
The \textsc{Union} operation destructively joins two sets.
The \textsc{Find} operation can be used to determine whether two elements,
$ e_i $ and $ e_j $, are in the same set.
Since the sets are disjoint, each element is in exactly one set.
Each set is identified by its \emph{representative element}, which can
be any of its elements.
The \textsc{Find} operation returns, for any element $ e_i $, the
representative element of the set to which $ e_i $ belongs.
Thus $ \textsc{Find}(e_i) = \textsc{Find}(e_j) $ if $ e_i $
and $ e_j $ are in the same set.

To solve the edge orientation problem, the edges of the mesh become the
abstract elements for the Union-Find algorithm, starting with a singleton
set for each edge. Then, for each cell, we join both pairs of opposite
edges, more precisely we join the sets to which those opposite edges
belong. When done, each set corresponds to a ``ribbon'', a set of edges
which determine each other's orientation.  However, we still need to
traverse the ribbon to determine the permitted relative orientations between
edges.  We will avoid this traversal by extending the Union-Find algorithm,
but let us first briefly present the usual representation of sets and the
algorithms for \textsc{Find} and \textsc{Union}.

Each element $ u $ links to its \emph{parent element} $ p(u) $.
For representative elements, $ p(u) = u $.
For all other elements, $ p(u) $ is another element in the same set,
and following the parent links must eventually lead to the
representative element of the set.
Alg.~\ref{alg:find} and Alg.~\ref{alg:union} show the implementation of the
\textsc{Find} and \textsc{Union} operations, respectively.

\begin{algorithm}
  \caption{\textsc{Find} operation for disjoint sets}
  \label{alg:find}
  \begin{algorithmic}[1]
    \Function {Find}{$ u $}
    \If {$ p(u) = u $}
      \State \Return $ u $
    \Else
      \State \Return \Call {Find}{$ p(u) $}
    \EndIf
    \EndFunction
  \end{algorithmic}
\end{algorithm}

\begin{algorithm}
  \caption{\textsc{Union} operation for disjoint sets}
  \label{alg:union}
  \begin{algorithmic}[1]
    \Procedure {Union}{$ u $, $ v $}
    \State $ r_u \gets $ \Call {Find}{$ u $}
    \State $ r_v \gets $ \Call {Find}{$ v $}
    \If {$ r_u \neq r_v $}
      \Comment{$ r_v $ becomes the representative element of the united set}
      \State $ p(r_u) \gets r_v $
    \EndIf
    \EndProcedure
  \end{algorithmic}
\end{algorithm}

\textsc{Find} simply traverses the parent links until it reaches the
representative element.  A variation of \textsc{Find} also reassigns the parent
links of each element on this path directly to the representative element --
this technique is called \emph{path compression}.
\textsc{Union} first calls \textsc{Find} to look up the representative elements
of both sets. If the two sets are indeed different, then it assigns the
parent link of one of the representative elements to point to the other
representative element, thus effectively merging the sets.
Depending on the exact variant of the Union-Find algorithm \cite{Patwary2010},
there is generally a rule to determine which one of $ r_u $ and $ r_v $ remains
a representative element, but we ignore that detail for this discussion.
The textbook version of the Union-Find algorithm \cite{tarjan1975efficiency}
uses \emph{rank based} merging in \textsc{Union} and path compression in
\textsc{Find}, and it is proven that $ m $ \textsc{Find} and $ n $
\textsc{Union} operations execute in $ O( (m + n) \alpha(m, n) ) $ time, where
$ \alpha(m, n) $, the functional inverse of Ackermann's function, is so
slow-growing that for all practical purposes it can be considered constant.

To address the edge orientation problem, we attach a binary value
$ \omega(u) $ to each parent link, which describes the
permitted relative orientation between the edge $ u $ and its parent
edge $ p(u) $.  Alg.~\ref{alg:find_orient} and
Alg.~\ref{alg:union_orient} describe \textsc{Find} and \textsc{Union}
with orientations, respectively.
The changes include the maintenance of orientation, and Möbius strip
detection in \textsc{Union} when one tries to connect a ribbon to
itself with the inconsistent orientation.

\begin{algorithm}
  \caption{\textsc{Find} operation for disjoint sets \emph{with orientation}}
  \label{alg:find_orient}
  \begin{algorithmic}[1]
    \Function {Find}{$ u $}
    \If {$ p(u) = u $}
      \State \Return $ u $, $ \uparrow $
    \Else
      \State $ r_u, o_p \gets $ \Call {Find}{$ p(u) $}
      \State \Return $ r_u, \, \omega(u) \text{ XOR } o_p $
    \EndIf
    \EndFunction
  \end{algorithmic}
\end{algorithm}

\begin{algorithm}
  \caption{\textsc{Union} operation for disjoint sets \emph{with orientation}}
  \label{alg:union_orient}
  \begin{algorithmic}[1]
    \Procedure {Union}{$ u $, $ v $, $ o $}
    \State $ r_u, o_u \gets $ \Call {Find}{$ u $}
    \State $ r_v, o_v \gets $ \Call {Find}{$ v $}

    \If {$ r_u = r_v $}
      \If {$ o_u \text{ XOR } o \text{ XOR } o_v ={} \downarrow $}
        \Comment{Möbius strip found}
        \State \textbf{abort}
      \EndIf
    \Else
      \Comment{$ r_v $ becomes the representative element of the united set}
      \State $ p(r_u) \gets (r_v, \, o_u \text{ XOR } o \text{ XOR } o_v ) $
    \EndIf
    \EndProcedure
  \end{algorithmic}
\end{algorithm}

This extension can be easily applied to most of the relevant variants
of the Union-Find algorithm, including parallel ones.
Alg.~\ref{alg:general} describes the algorithm to solve the edge
orientation problem using the Union-Find algorithm extended with
orientations. While Alg.~\ref{alg:general} itself does not feature any
parallel primitives, it will solve the edge orientation problem in
parallel if a parallel variant of the Union-Find algorithm
is employed. In that case each parallel process only iterates over its
\emph{local} cells and edges.

\begin{algorithm}
  \caption{Edge orientations using the extended Union-Find algorithm}
  \label{alg:general}
  \begin{algorithmic}[1]
    \ForAll {$ c \in \text{\(local\) cells} $}
      \State $ e_1, e_2, e_1', e_2' \gets \text{edges of } c \text{, each consecutive pair sharing a vertex} $
      \State $ o_1 \gets \text{permitted relative orientation between } e_1 \text{ and } e_1' $
      \State $ o_2 \gets \text{permitted relative orientation between } e_2 \text{ and } e_2' $
      \State \Call {Union}{$ e_1 $, $ e_1' $, $ o_1 $}
      \State \Call {Union}{$ e_2 $, $ e_2' $, $ o_2 $}
    \EndFor
    \ForAll {$ e \in \text{\(local\) edges} $}
      \State $ r_e, o_e \gets $ \Call {Find}{$ e $}
      \State $ \mathtt{orientations}[e] \gets o_e $
    \EndFor
  \end{algorithmic}
\end{algorithm}

The advantage of this approach is that we can easily re-use existing
work on parallelising the Union-Find algorithm to solve the edge
orientation problem.
Several attempts have been made to parallelise that algorithm both
on shared memory computers \cite{Cybenko1988, Anderson1991, Bader2005,
  Patwary2012}, and on distributed memory systems \cite{Cybenko1988,
  Manne2010, Harrison2011, Iverson2015}.
Since Firedrake is meant to scale on today's supercomputers, we are
only interested in distributed memory parallelism.
Cybenko et al.\@ \cite{Cybenko1988} attempt both shared memory and
distributed memory parallelisation, but their distributed memory
algorithm suffered from slowdown with increasing number of processors
and constant problem size.
Manne and Patwary \cite{Manne2010} propose a distributed memory
parallel algorithm, which they demonstrate to scale up to 40
processors for certain graphs. This is, however, far from the scale
of supercomputers on which Firedrake is supposed to run.
Harrison et al.\@ \cite{Harrison2011} achieve good scaling for their
purposes, which is connected component labeling for distributed
memory visualisation tools. However, the final step of their algorithm
uses all-to-all communication to merge the local components, which we
expect to become a serious bottleneck for our purposes, since
distributed quadrilateral meshes produce numerous local ribbon
segments and connections between them.
Iverson et al.\@ \cite{Iverson2015} compare five parallel algorithms
for connected component labeling on distributed memory systems, both
theoretically and experimentally. The approach they call \emph{label
  propagation} is essentially equivalent to the parallelisation of
the edge orientation algorithm implemented in Firedrake.
Although it was not always the fastest in their evaluations,
it did not suffer from the explosive use of memory which affected some
other algorithms.

\section{Conclusion \& Outlook}
\label{sec:conclusion}

We have proposed a distributed mesh parallel algorithm, extending the
serial quadrilateral edge orientation algorithm in deal.II.
We have proven that there are consistent edge orientations for any
quadrilateral mesh of an orientable 2-dimensional manifold.
Our novel analysis establishes a link between the well-known
\emph{Union-Find} algorithm and assigning orientations to the edges of
a quadrilateral mesh, thus existing parallelisation approaches for
the Union-Find algorithm can be easily adapted to construct further
parallel algorithms for mesh orientations.

\bibliography{references}

\end{document}